\documentclass[conference]{IEEEtran}

\usepackage{amsmath, amsthm, amssymb}
\usepackage{latexsym}
\usepackage{graphicx}
\usepackage{verbatim}
\usepackage{mathrsfs}
\usepackage{float} 
\usepackage[lofdepth, lotdepth]{subfig}
\usepackage{array}
\newcolumntype{P}[1]{>{\centering\arraybackslash}p{#1}}
\usepackage{url}
\usepackage{algorithm}
\usepackage{algpseudocode}
\usepackage{cite}
\usepackage[table]{xcolor}
\usepackage{enumerate}
\usepackage{enumitem}
\usepackage{eucal}
\usepackage{csquotes}
\usepackage{stmaryrd}
\usepackage{mathtools}
\usepackage{pgf}
\usepackage{tikz}
\usetikzlibrary{arrows,automata}
\usepackage[latin1]{inputenc}
\usetikzlibrary{automata,positioning}
\usepackage[hidelinks, bookmarks=false]{hyperref}

\usepackage{mathtools}
\DeclarePairedDelimiter\ceil{\lceil}{\rceil}
\DeclarePairedDelimiter\floor{\lfloor}{\rfloor}

\theoremstyle{plain}
\newtheorem{theorem}{Theorem}

\newtheorem{lemma}{Lemma}

\theoremstyle{definition}
\newtheorem{definition}{Definition}

\newtheorem{remark}{Remark}

\newcommand{\B}{{\mathcal B}}
\newcommand{\C}{{\mathcal C}}

\DeclareMathAlphabet{\mathbfsl}{OT1}{ppl}{b}{it} 

\newcommand{\bS}{{\mathbfsl{S}}}

\newcommand{\bL}{\mathbfsl{L}}
\newcommand{\bA}{{\mathbfsl A}}

\newcommand{\by}{{\mathbfsl y}}

\newcommand{\bw}{{\mathbfsl{w}}}

\newcommand{\bx}{{\mathbfsl{x}}}
\newcommand{\bz}{{\mathbfsl{z}}}

\newcommand{\enc}{\textsc{Enc}}
\newcommand{\dec}{\textsc{Dec}}

\renewcommand{\ge}{\geqslant}
\renewcommand{\le}{\leqslant}

\newcommand{\et}{{\emph{et al.}}}

\IEEEoverridecommandlockouts

\def\BibTeX{{\rm B\kern-.05em{\sc i\kern-.025em b}\kern-.08em
    T\kern-.1667em\lower.7ex\hbox{E}\kern-.125emX}}

\begin{document}

\pagestyle{empty}


\title{From One-Dimensional Codes to Two-Dimensional Codes: A Universal Framework for the Bounded-Weight Constraint\\[-3mm]}



\author{\IEEEauthorblockN{
Viet Hai Le\IEEEauthorrefmark{1},
Thanh Phong Pham\IEEEauthorrefmark{2},
Tuan Thanh Nguyen\IEEEauthorrefmark{1},
Kui Cai\IEEEauthorrefmark{1},
and Kees A. Schouhamer Immink\IEEEauthorrefmark{3}}\\[-3mm]

\IEEEauthorblockA{
\IEEEauthorrefmark{1}
Singapore University of Technology and Design, Singapore 487372\\
\IEEEauthorrefmark{2}%
Nanyang Technological University, Singapore 639798\\
\IEEEauthorrefmark{3}%
Turing Machines Inc, Willemskade 15d, 3016 DK Rotterdam, The Netherlands\\
[-4mm]
}
\thanks{This work was supported by the SUTD Kickstarter Initiative (SKI) Grant 2021\_04\_05 and the Singapore Ministry of Education Academic Research Fund Tier 2 T2EP50221-0036.}
}
\maketitle

\hspace{-3mm}\begin{abstract}
Recent developments in storage- especially in the area of resistive random access memory (ReRAM)- are attempting to scale the storage density by regarding the information data as two-dimensional (2D),  instead of one-dimensional (1D). Correspondingly, new types of 2D constraints are introduced into the input information data to improve the system reliability. While 1D constraints have been extensively investigated in the literature, the study for 2D constraints is much less profound. Particularly, given a constraint $\mathcal{F}$ and a design of 1D codes whose codewords satisfy $\mathcal{F}$, the problem of constructing efficient 2D codes, such that every row and every column in every codeword satisfy $\mathcal{F}$, has been a challenge.
\vspace{1mm}


This work provides an efficient solution to the challenging coding problem above for the binary bounded-weight constrained codes that restrict the maximum number of $1$'s (called {\em weight}). Formally, we propose a universal framework to design 2D codes that guarantee the weight of every row and every column of length $n$ to be at most $f(n)$ for any given function $f(n)$. We show that if there exists a design of capacity-approaching 1D codes, then our method also provides capacity-approaching 2D codes for all $f=\omega(\log n)$.

\end{abstract}

\section{Introduction}
Binary weight-constrained codes are useful in various applications in the fields of computer science, information security, communications, and data storage systems. For example, the {\em constant-weight} (CW) codes, whose codewords share the same {\em Hamming weight}, have been used in several applications including frequency hopping in GSM networks \cite{mambou:2017,smith:2006}, optical and magnetic recording \cite{vitaly:2014, immink:book}. Particularly, the {\em balanced codes}, in which every codeword has an equal number of
$0$'s and $1$'s, have many applications in both classical and modern communications and data storage systems such as digital optical disks \cite{immink:book,lei84}, data synchronization \cite{ofe1990, alon:1988}, and the emerging DNA-based data storage \cite{tu:2024, thanh:dna, hec:2019, dube:2019}. In this work, we study the {\em bounded-weight} constrained codes, which restrict the maximum number of $1$'s in every codeword and have been used widely to tackle several challenging communications and storage systems issues, especially in the area of resistive random access memory (ReRAM).  
\vspace{0.5mm}

We first describe briefly certain applications of the bounded-weight constraint, including 1D codes and 2D codes. For instance, such 1D codes have been used to prevent energy outage at a receiver having finite energy storage capability in simultaneous energy and information transfer \cite{bound1d3, tandon:2016,HM:2021}, or to prevent high temperatures in electronic devices \cite{chee:TA2024, chee:TA2023, immink:survey}. On the other hand, the 2D bounded-weight codes, which restrict the maximum number of $1$'s in every row and every column of an array codeword, have been suggested as an effective solution to reduce the {\em sneak path interference}, a fundamental and challenging problem in ReRAM \cite{ordentlich2012low, nguyen2021efficient, nguyen2023locally,nguyen2023two,cassuto2013sneak,cassuto2016information}. 
\vspace{0.5mm}

While 1D constraints have been extensively investigated, the study for 2D constraints is much less profound. 
For instance, for the 1D bounded-weight constraint, there are several efficient prior-art coding methods for designing 1D codes with optimal redundancy or almost optimal (up to a constant number of redundant bits away from the optimal redundancy), particularly including the {\em enumeration coding technique} \cite{immink:book}, {\em Knuth's balancing technique} \cite{knuth:1986,bound1d3, vitaly:2014},  the classic {\em polarity bit code} \cite{immink:book,bound1d3}, and the {\em sequence replacement technique} \cite{bound1d3, gabry2018, gabrys:2020}. Among them, the enumeration coding technique has been considered an efficient solution providing optimal redundancy, despite its disadvantages of high complexity and high error propagation \cite{immink:book, chee:tandem, thanh:dna, xhe:2022}. Nevertheless, using such a technique to handle 2D constraints has remained a challenge.  
In this work, we are interested in the following coding problem: 
\begin{displayquote}
Given a constraint $\mathcal{F}$ (or generally, a finite set of constraints $\mathcal{F}$) and an efficient construction of 1D codes whose codewords satisfy $\mathcal{F}$, how to extend the construction to 2D codes, i.e. every row and every column in every 2D codeword satisfy $\mathcal{F}$?
\end{displayquote}

Our main contribution is to provide an efficient solution to the challenging problem above when $\mathcal{F}$ is a bounded-weight constraint. Formally, we propose a universal framework to design 2D codes that guarantee the weight of every row and every column of length $n$ to be at most $f(n)$ for any given function $f(n)$, through efficient encoding and decoding algorithms. We show that if a design of capacity-approaching 1D codes exists, then our encoder provides capacity-approaching 2D codes for all $f=\omega(\log n)$. Our results extend the study of literature works that specifically require $f(n)=n/2$ as proposed by Ordentlich and Roth \cite{ordentlich2012low}, or $f(n)=pn$ for a constant $p\in(0,1)$ as proposed by Nguyen \et{} \cite{nguyen2023two}.
\vspace{0.5mm}

\section{Preliminary}\label{sec:prelim}



Given two sequences $\bx=x_1\ldots x_{n_1}$ and $\by=y_1\ldots y_{n_2}$, the \emph{concatenation} of the two sequences is defined by $\bx \by \triangleq x_1\ldots x_{n_1} y_1\ldots y_{n_2}$. 
\vspace{1 mm}

For a binary sequence $\bx$, we use $\operatorname{wt}(\bx)$ to denote the {\em weight} of $\bx$, i.e. the number of ones in $\bx$. In addition, we use $\bar{\bx}$ to denote the complement of $\bx$. For example, if $\bx = 110011$, then $\operatorname{wt}(\bx)=4$ and $\overline{\bx}=001100$. 

\begin{definition}\label{def-swap}

Given $\by=y_1y_2\ldots y_n, \bz=z_1z_2\ldots z_n$. For $1\le t\le n$, we use ${\rm Swap}_t(\by,\bz)$, ${\rm Swap}_t(\bz,\by)$ to denote the sequences obtained by swapping the first $t$ bits of $\by$ and $\bz$, i.e.
\begin{align*}
{\rm Swap}_t(\by,\bz) &= z_1z_2\ldots z_t y_{t+1}y_{t+2}\ldots y_n, \text{ and } \\
{\rm Swap}_t(\bz,\by) &= y_1y_2\ldots y_t z_{t+1}z_{t+2}\ldots z_n. 
\end{align*}
\end{definition}

\noindent{\em Big-O, little-o, big-omega and little-omega functions}. 

Given two real functions $f$ and $g$, we say that $f$ is big-O of $g$ and write $f=O(g)$, if there exists a positive real number $M$ and a real number $x_0$ such that $|f(x)|\le M g(x)$ for all $x\ge x_0$. We say that $f$ is little-o of $g$ and write $f=o(g)$ if $\lim_{x\to \infty} f(x)/g(x)=0$. 

On the other hand, $f$ is big-omega of $g$, denoted by $\Omega(g)$, if there exists a positive real number $M$ and a real number $x_0$ such that $f(x)\ge M |g(x)|$ for all $x\ge x_0$. We say $f$ is little-omega of $g$, denoted by $f=\omega(g)$, if $\lim_{x\to \infty} f(x)/g(x)=\infty$, or equivalently, $\lim_{x\to \infty} g(x)/f(x)=0$.

\begin{definition}
Given a function $f$, a sequence $\bx\in\{0,1\}^n$ is called $f$-bounded if its weight satisfies that $\operatorname{wt}(\bx)\le f(n)$. 
\end{definition}

Let \(\bA_{n}\) denote the set of all \(n \times n\) binary arrays. For an array \(A \in \bA_{n}\), the $i$th row is denoted by \(A_{i}\), the $j$th column is denoted by \(A^{j}\), and the index bit at the (row, column) location $(i,j)$ is denoted by $a_{i,j}$. Note that an $n\times n$ binary array $A$ can be viewed as a binary sequence of length $n^2$. We define $\Phi(A)$ as a binary sequence of length $n^2$ where bits of the array $A$ are read row by row. Clearly, $\Phi$ is a one-to-one mapping.
Similarly, we also write $\operatorname{wt}(A)$ to represent the number of ones in $A$. 
Given $n>0$ and a function $f$, where $0\le f(n)\le n$ for all $n$, we are interested in the set ${\B}(n;f) \subseteq \bA_{n}$, where:
\begin{align*}
    {\B}(n;f)=\Big\{ A \in \bA_{n}: &\operatorname{wt}(A_i) \le f(n) \text{ for } 1\le i\le n, \\ 
    &\operatorname{wt}(A^j) \le f(n) \text{ for } 1\le j\le n \Big\}.
\end{align*}
Given a function $f$, the channel capacity\footnote{In this work, for simplicity, we use the notation ``$\log$" without the base to refer to the logarithm of base two.} is defined by 
\begin{equation*}
    {\bf cap}_f=\lim_{n\to \infty}\frac{\log|{\B}(n;f)|}{n^2} \le 1.
\end{equation*}

A set $\C \subseteq {\B}(n;f)$ is called a 2D $f$-bounded code. The {\em rate} of a code $\C$, denoted by ${\bf r}_{\C}$, is measured by the value ${\bf r}_{\C}=(1/n^2)\log |\C|$, and the {\em redundancy} of a code $\C$, denoted by ${\bf Re}_{\C}$, is measured by the value ${\bf Re}_{\C}=n^{2} - \log|\C|$ (bits). A code $\C$ is a {\em capacity-approaching} 2D $f$-bounded code if:
\begin{equation*}
\lim_{n\to \infty} {\bf r}_{\C}={\bf cap}_f.
\end{equation*} 


In this work, we design efficient coding methods that encode (decode) arbitrary binary data to (from) a code $\C\subseteq {\B}(n;f)$, and we aim to achieve capacity-approaching 2D codes. 

\begin{definition}\label{def2}
The map $\enc_f: \{0,1\}^{k} \rightarrow \{0,1\}^{n \times n}$ is a 2D $f$-bounded encoder if for all $\bx\in\{0,1\}^k$ we have $\enc_f(\bx) \in {\B}(n;f)$, and there exists a corresponding $f$-bounded decoder $\dec_f: \{0,1\}^{n \times n} \rightarrow \{0,1\}^{k}$ such that $\dec_f \circ \enc_f(\bx) = \bx$. 
\end{definition}

We observe that the set $\C=\{\enc_f(\bx): x\in \{0,1\}^{k}\}$ forms a 2D $f$-bounded code. The rate of $\C$ (or the rate of the encoder) is $k/n^{2}$, and the redundancy is $n^{2} - k$ (bits). 

\subsection{Previous Works}\label{previous-work} 

In \cite{ordentlich2012low}, for arrays of size $n \times n$, Ordentlich and Roth considered $f(n)=n/2$ and presented efficient encoders that used at most $2n-1$ redundant bits. In such a case, we have ${\bf cap}_f=1$ for $f(n)=n/2$. In \cite{nguyen2023two}, Nguyen \et{} studied $f(n)=pn$ for arbitrary constant $p\in(0,1)$, and it was shown in both \cite{nguyen2023two,ordentlich2000two} that ${\bf cap}_f<1$ for any constant $p<1/2$. 
\vspace{1mm}

\begin{figure}[H]
    \centering
    \includegraphics[width=0.35\textwidth]{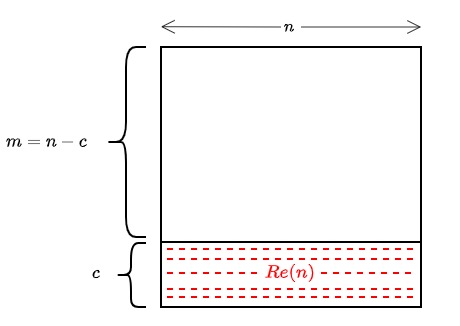}
   \caption{The common steps of our work and the encoding algorithm in \cite{nguyen2023two}. The efficiency of both encoding schemes depends on the value of $c$; the smaller the value of $c$, the lower the redundancy required.}\label{fig-encode}
    \end{figure}

We briefly describe the construction in \cite{nguyen2023two} for $f(n)=pn$ for a constant $p\in(0,1)$. Crucial to the encoding algorithm in \cite{nguyen2023two} are two values $m,c$, defined according to the given values of $n,p$ as follows:   
\begin{align}
    c &\ge \ceil{1/p}\Big( \ceil{\log n}+1 \Big), \text{ and } cp \text{ is an integer}, \label{c-bound1}\\
    m &= n-c. \nonumber
\end{align}
In particular, 
\begin{itemize}
    \item The first $m$ rows are simultaneously encoded to be $f$ bounded. The array is then divided into two subarrays of equal size, and the problem is reduced to enforce the constraint over two subarrays. The process is repeated until each subarray is a single column. For each subproblem, the authors \cite{nguyen2023two} {\em swap} the bits between two subarrays (according to Definition~\ref{def-swap}) until both subarrays satisfy the weight constraint. The number of swapped bits $t$ during each step is referred to as the {\em swapping index}. 
    \item The last $c$ rows are used to record all the swapping indices when encoding the first $m$ rows (referred to as $Re(n)$ in Figure~\ref{fig-encode}). 
\end{itemize}

Observe that one can use the encoding scheme in  \cite{nguyen2023two} to restrict the weight of every row and every column to be at most $f(n)$ by setting $p = f(n)/n$. Consequently, the parameter $c\in \mathbb{N}$ can be chosen as:
\begin{equation}
    c \geq \ceil{n/f(n)}\Big( \ceil{\log n} + 1 \Big),\text{where } \frac{cf(n)}{n} \in \mathbb{N}
    \label{eq:old}
\end{equation}

Determining a suitable value of $c$ is then non-trivial. When $f(n)=pn$ as in ~\eqref{c-bound1}, $p=O(1)$, and hence, $c_{~\eqref{c-bound1}}=O(\log n)$. 

\subsection{Our Novel Contribution}
In this work, we propose an efficient coding scheme that works for any arbitrary function $f(n)$. 

$(i)$ Our first contribution is to provide an explicit value of $c$ as follows:
\begin{equation}
    c = \ceil{n/f(n)} \times \ceil{\log n + 6}.
    \label{eq:new}
\end{equation}

Note that the bound in Equation~\eqref{eq:old} often results in a higher value of $c$ due to the additional requirement that $cf(n)/n \in \mathbb{N}$. For example, consider $f(n) = \frac{n}{2} - \log n$, where $n$ is a power of $2$. We then require
\begin{equation*}
    2c_{~\eqref{eq:old}} \cdot \frac{\log n}{n} \in \mathbb{N}, \text{ which implies: } c_{~\eqref{eq:old}} = \Omega\Big(\frac{n}{\log n} \Big).
\end{equation*}
On the other hand, we obtain $c_{~\eqref{eq:new}}=O(\log n)$, which is significantly smaller than $c_{~\eqref{eq:old}}$, and therefore improves overall redundancy. 

$(ii)$ We extend the swapping algorithm in \cite{nguyen2023two} so that it works for any arbitrary function $f$. A trivial solution is that, if there is always a subarray $\bS$ violating the weight constraint after some swapping step, one can simply flip some bits `$1 \to 0$' in $\bS$ to further reduce the weight in $\bS$. This flipping step requires additional redundancy (roughly about $\log |\bS|$ bits) to record the flipping bits' locations (or index). In this work, we design the swapping algorithm so that in the worst-case scenario, when there is always a subarray $\bS$ violating the weight constraint, the additional flipping step proceeds at exactly one index. Surprisingly, such a flipping operation (if needed) is always at the swapping index (refer to Lemma~\ref{main_lemma}), and hence, we only require two redundant bits to indicate a flipping action.

Our coding scheme is described in Section III.

% ...

% ...


% ...


% ...

\section{A Universal Framework for the 2D Bounded-Weight Constraint}\label{mainresult}

Suppose that there exists an 1D encoding map $\phi: \{0,1\}^k \to \{0,1\}^n$ such that for all $\bx\in \{0,1\}^k$, the sequence $\by=\phi(\bx)\in\{0,1\}^n$ is $f$-bounded, i.e. ${\rm wt}(\by) \le f(n)$. 
For arrays of size $n\times n$, we set 
\begin{equation*}
c = \Big\lceil\frac{n}{f(n)}\Big\rceil \lceil\log n + 6\rceil, \text{ and } m=n-c.
\end{equation*}

Our encoding algorithm works if $c<n$ or $f(n) \gtrsim \log n+6$.

A high-level description of our encoder is as follows. 
\begin{itemize}
\item {\bf (Using 1D codes)} The encoder uses the map $\phi$ to encode the first $m$ rows to be $f$-bounded, and obtains an array $A$ of size $m\times n$, whose weight is at most $\floor{mf(n)}$. 
\item {\bf (Divide-and-conquer procedure)} The encoder uses Lemma~\ref{main_lemma} to encode $A$ to an array $B$ of size $m\times n$ where each row remains as a $f$-bounded sequence, and the weight of each column is at most $\floor{mf(n)/n}$.
\item {\bf (Encoding the redundancy)} The encoder encodes the redundancy used in the divine-and-conquer procedure into an array $C$ of size $c \times n$ such that every row of $C$ is $f$-bounded, while every column weights is at most $\floor{cf(n)/n}$. Finally, the encoder outputs the concatenation of $B$ and $C$, which is an array $D$ of size $n\times n$. We observe that the weight of every column of $D$ is at most 
\begin{align*}
\floor{mf(n)/n} + \floor{cf(n)/n} &\le \floor{mf(n)/n+cf(n)/n} \\
&=\floor{f(n)} \le f(n).
\end{align*}
Thus, the output array $D\in \B(n;f)$. 
\end{itemize}

The following results are crucial to the divide-and-conquer procedure. Lemma~\ref{swap} is immediate. 

\begin{lemma} \label{swap}
 Let $\by=y_{1} y_{2} \ldots y_{n}$ and $\bz=z_{1} z_{2} \ldots z_{n}$. Consider $1\le t_1, t_2\le n$ where $\operatorname{wt}(\operatorname{Swap}_{t_1}(\boldsymbol{y}, \boldsymbol{z})) = \alpha$ and $\operatorname{wt}(\operatorname{Swap}_{t_2}(\boldsymbol{y}, \boldsymbol{z})) = \beta$ with $\alpha < \beta$. For arbitrary integer $\gamma \in [\alpha,\beta]$, there exists an index $t$ between $t_1$ and $t_2$ such that $\operatorname{wt}(\operatorname{Swap}_{t}(\boldsymbol{y}, \boldsymbol{z})) = \gamma$.
\end{lemma}

\begin{lemma} \label{main_lemma}
Given $n>0$, $\by, \bz \in\{0,1\}^{n}$, where $\operatorname{wt}(\by) > \alpha$ and $\operatorname{wt}(\bz) < \alpha$. We then conclude:
\begin{enumerate}[label=\roman*)]
\item There exists $t$, $1 \leqslant t \leqslant n$,  such that  $\operatorname{wt}(\operatorname{Swap}_{t}(\bz,\by)) = \lfloor \alpha \rfloor + 1 $. Furthermore, if $t$ is the smallest index satisfying this condition, then the $t$-th position of $\operatorname{Swap}_{t}(\bz,\by)$ is $1$. 
\item Similarly, we have the following claim. There exists $t$, $1 \leqslant t \leqslant n$,  such that  $\operatorname{wt}(\operatorname{Swap}_{t}(\by,\bz)) = \lfloor \alpha \rfloor $. Furthermore, if $t$ is the smallest index satisfying this condition, then the $t$-th position of $\operatorname{Swap}_{t}(\bz,\by)$ is $1$.
\end{enumerate}
\end{lemma}

\begin{proof}
Due to space constraints, we only show the correctness of the first statement. The second statement can be proved similarly. Since $\operatorname{wt}(\by) > \alpha$ and $\operatorname{wt}(\bz) < \alpha$, by using Lemma \ref{swap} with $\operatorname{wt}(\operatorname{Swap}_{0}(\bz,\by)) = \operatorname{wt}(\bz) \leqslant \floor{\alpha}$, and $\operatorname{wt}(\operatorname{Swap}_{n}(\bz,\by)) = \operatorname{wt}(\by) > \alpha $, or equivalently, $\operatorname{wt}(\operatorname{Swap}_{n}(\bz,\by)) \ge \floor{\alpha}+1$, we conclude that there exists an index $t$, $1 \leqslant t \leqslant n$, that $\operatorname{wt}(\operatorname{Swap}_t(\bz,\by)) = \floor{\alpha} + 1$. 

Let $t$ be the minimum index, which satisfies that $\operatorname{wt}(\operatorname{Swap}_{t}(\boldsymbol{z},\boldsymbol{y})) = \floor{\alpha} + 1$. Set $\bw = \operatorname{Swap}_{t}(\boldsymbol{z},\boldsymbol{y})$. We next show that $w_t = 1$. We prove this by contradiction. Suppose that, on the other hand, we have $w_t=0$. 
Note that $\bw = \operatorname{Swapt}_{t}(\bz, \by) = y_1 \ldots y_{t-1} y_t  z_{t+1} ... z_n$. If $w_t = 0$, then $y_t = 0$. We then have two following cases. 

\begin{itemize}
    \item If $z_t = 0$. We observe that $\operatorname{Swap}_{t-1}(\bz, \by) = y_1 \ldots  y_{t-1} z_t z_{t+1} ... z_n$ has the same weight with $\operatorname{Swap}_{t}(\bz, \by)$. Since $t -1 < t$, we have a contradiction.

    \item If $z_t = 1$. Since $y_t = 0$, it is easy to verify that $\operatorname{wt}(\operatorname{Swap}_{t-1}(\bz,\by)) = \operatorname{wt}(\operatorname{Swap}_{t}(\bz,\by)) + 1 = \floor{\alpha} + 2$. Moreover, we have $\operatorname{wt}(\operatorname{Swap}_{0}(\bz,\by)) = \operatorname{wt}(z) \leqslant \floor{\alpha} $. According to Lemma \ref{swap}, there exists an index $t'$, where $ 0 \leqslant t' \leqslant t - 1$, such that $\operatorname{wt}(\operatorname{Swap}_{t'}(\bz,\by)) = \floor{\alpha}+ 1$. Again, we have a contradiction. 
\end{itemize}
    
In conclusion, if $t$ is the smallest index satisfying the weight condition, then the $t$-th position of $\operatorname{Swap}_t(\bz,\by)$ must be $1$.
\end{proof}

\subsection{The Divide-and-Conquer Procedure: Main Idea}\label{main_section}
Suppose that at some encoding step $i$, we have an array $\bS$ containing $n_i$ columns, and the overall weight $\operatorname{wt}(\bS)\le \floor{n_i\alpha}$ for some $\alpha>0$. In the initial step $i=0$, we have $n_0=n$. 
\begin{itemize}
\item {\bf (Case 1)} If $n_i$ is even, the encoder divides $\bS$ into two subarrays $L_{\bS}$ and $R_{\bS}$ of equal size, each containing exactly $n_i/2$ columns. The encoder then uses Lemma~\ref{main_lemma} to swap the bits between two subarrays until the weight of each subarray is at most $\floor{(n_i/2)\alpha}$. The number of swapped bits, called {\em the swapping index}, will be recorded for the decoding procedure. In addition, a flipping action to convert a `1' to a `0' in one of the subarrays may be needed, and we refer it to as the {\em flipping index}. 

\item {\bf (Case 2)} On the other hand, if $n_i$ is odd, the encoder divides $\bS$ into: $L_{\bS}$ consisting of the first $(n_i+1)/2$ columns, and $R_{\bS}$ consisting of the remaining $(n_i-1)/2$ columns. To use Lemma~\ref{main_lemma} (which only applies to two sequences of equal length), the encoder removes one column from $L_{\bS}$ (the column's index will be recorded, and referred to as the {\em excluding index}) and proceeds to the swapping algorithm. The goal is to ensure that the weight of $L_{\bS}$ is at most $\floor{((n_i+1)/2)\alpha}$, while the weight of $R_{\bS}$ is at most $\floor{((n_i-1)/2)\alpha}$.
\end{itemize}

\noindent{\bf Output of the divide-and-conquer procedure.} We observe that at every step $j$, for any new subarray consisting of $n_j$ columns, its weight is enforced to be at most  $\floor{n_j\alpha}$. The procedure terminates when each subarray contains exactly one column, and the weight is at most $\floor{\alpha}$. Therefore, if the overall weight in the initial step of an array is at most $\floor{mf(n)}$, then in the final step, the weight of every column is at most $\floor{mf(n)/n}$. There are three types of indices: 
\begin{itemize}
    \item \textbf{Swapping index} ($t_i$): represents the number of swapped bits between $L_{S}$ and $R_{S}$. Recall that all created subarrays have exactly $m$ rows. If the array $\bS$ consists of $n_i$ columns, we need $\ceil{\log(mn_i/2)}$ bits to record $t_i$. 
    \item \textbf{Flipping index} ($\tau_i$): indicates the need to transform a `1' to a `0' in one of the subarrays. Here, $\tau_i\in \{0,1,2\}$. Particularly, $\tau_i = 0$ indicates that no flipping action is needed, $\tau_i = 1$ indicates that such a flipping action is done in the left subarray $L_{\bS}$, while $\tau_i = 2$ indicates that the flipping action is done in the right subarray $R_{\bS}$. Hence, two redundant bits are sufficient to record $\tau_i$. 

    \item \textbf{Excluding index} ($\Gamma_i$): only be used when $n_i$ is odd. In other words, if $n$ is a power of $2$, then $\Gamma_i=0$ for all $i$. This index indicates the column in $\bL_{\bS}$ that would be excluded before we swap the bits. It is easy to see that we need $\ceil{\log(n_i+1)/2}$ bits to record $\Gamma_i$.
\end{itemize}

\subsection{The Divide-and-Conquer Procedure: Detailed Steps}




Consider an array $\bS$ containing $n_i$ columns, $\operatorname{wt}(\bS) \leqslant \floor{n_i\alpha}$. 

\noindent{\bf Case 1:}  $n_i$ is even.  Set $L_{\bS}$ be an array consisting of the first $n_i/2$ columns of $\bS$ and $R_{\bS}$ consists of all remaining $n_i/2$ columns. Recall that the number of rows in two subarrays is $m=n-c$. In the trivial case, if $\operatorname{wt}(L_{\bS}), \operatorname{wt}(R_{\bS})$ $\leqslant \floor{(n_i/2)\alpha}$, then no swap is needed and the swapping index is $t_i=0$. In addition, we have $\tau_i=0, \Gamma_i=0$.  
\vspace{1mm}

On the other hand, if one subarray has weight strictly more than $\floor{(n_i/2)\alpha}$. W.l.o.g, suppose that ${\rm wt}(L_{\bS})>(n_i/2)\alpha$ while ${\rm wt}(R_{\bS})< (n_i/2)\alpha$. We then swap the bits in $L_{\bS}$ and $R_{\bS}$ in the order from left to right, and column by column, respectively (see Figure~\ref{example}). Observe that such a swapping iteration does not change the weight in each row of $\bS$. According to Lemma~\ref{main_lemma}, there exists an index $t_i$, $1 \leqslant t_i \leqslant mn_i/2$, such that 
\begin{equation*}
\operatorname{wt}(\operatorname{Swap}_{t_i}(R_{\bS}, L_{\bS})) = \floor{(n_i/2)\alpha} + 1,
\end{equation*}
and the $t_i$-th position of $\operatorname{Swap}_{t_i}(R_{\bS}, L_{\bS})$, is equal to $1$. Moreover, for such an index $t_i$, it is easy to verify that 
\begin{equation*}
\operatorname{wt}(\operatorname{Swap}_{t_i}(L_{\bS}, R_{\bS}))\le  \floor{(n_i/2)\alpha},
\end{equation*} 
since the total weight of these two subarrays is bounded above by $\floor{n_i\alpha}$. Next, the encoder flips the $t_i$-th position of $\operatorname{Swap}_{t_i}(R_{\bS}, L_{\bS})$ from $1$ to $0$ (i.e. we have the corresponding flipping index $\tau_i=2$). Thus, the weight of the two subarrays after the swapping step and the flipping step is at most $\floor{(n_i/2)\alpha}$. Similarly, if ${\rm wt}(R_{\bS})>(n_i/2)\alpha$, then the flipping action is done over $\operatorname{Swap}_{t_i}(L_{\bS}, R_{\bS})$ and we have $\tau_i=1$. 

In conclusion, the swapping index $t_i$, the flipping index $\tau_i \in\{0,1,2\}$ are recorded for the decoding procedure. 
\vspace{1mm}

\noindent{\bf Case 2:}  $n_i$ is odd, i.e. $n_i=2n_i'+1$. Let $L_{\bS}$ be an array consisting of the first $(n_i'+1)$ columns, and $R_{\bS}$ be an array consisting of the remaining $n_i'$ columns. Again, if $\operatorname{wt}(L_{\bS}) \le \floor{(n_i'+1)\alpha}$ and $\operatorname{wt}(R_{\bS})\leqslant \floor{n_i'\alpha}$, then no action is needed. In such a case, we have $t_i=\tau_i=\Gamma_i=0$. 
\vspace{1mm}

\noindent {\bf Case 2a:} If $\operatorname{wt}(L_{\bS}) > (n_i'+1)\alpha$ while $\operatorname{wt}(R_{\bS})< n_i'\alpha$. The encoder identifies the column with {\em minimum weight} in $L_{\bS}$ and sets $L_{\bS}'$ be the new subarray consisting of the other $n_i'$ columns in $L_{\bS}$. Clearly, we have $\operatorname{wt}(L_{\bS}') > n_i'\alpha$. The excluding index $\Gamma_i$ is then $1\le \Gamma_i \le n_i'+1$.  

Next, the encoder proceeds to the swapping step for $L_{\bS}'$ and $R_{\bS}$, where $\operatorname{wt}(L_{\bS}') > n_i'\alpha$ and $\operatorname{wt}(R_{\bS})< n_i'\alpha$. Similarly, according to Lemma~\ref{main_lemma}, there exists an index $t_i$, $1 \leqslant t_i \leqslant mn_i'$: 
\begin{equation*}
\operatorname{wt}(\operatorname{Swap}_{t_i}(R_{\bS}, L_{\bS}')) = \floor{n_i'\alpha} + 1,  
\end{equation*}
and the $t_i$-th position of $\operatorname{Swap}_{t_i}(R_{\bS}, L_{\bS}')$ is equal to $1$. Moreover, we verify that with such an index $t_i$, the total weight of $\operatorname{Swap}_{t_i}(L_{\bS}', R_{\bS}')$ and the excluding column $\Gamma_i$ is at most 
\begin{equation*}
\floor{(2n_i'+1)\alpha} - \floor{n_i'\alpha} - 1 \le \floor{(n_i'+1)\alpha}.  
\end{equation*}
Next, the encoder flips the $t_i$-th position of $\operatorname{Swap}_{t_i}(R_{\bS}, L_{\bS}')$ from $1$ to $0$. Finally, at the end of step $i$, the encoder obtains two subarrays satisfying the weight constraint. 

Again, it records all the indices: $1\le t_i\le m(n_i-1)/2$, $\tau_i \in\{0,1,2\}$, and $1\le \Gamma_i \le (n_i+1)/2$.  
\vspace{1mm}

\noindent {\bf Case 2b:} If $\operatorname{wt}(L_{\bS}) < (n_i'+1)\alpha$ while $\operatorname{wt}(R_{\bS})> n_i'\alpha$. In this case, the encoder identifies the column with {\em maximum weight} in $L_{\bS}$ and sets $L_{\bS}'$ be the new subarray consisting of the other $n_i'$ columns in $L_{\bS}$. Clearly, we have $\operatorname{wt}(L_{\bS}') < n_i'\alpha$. Similarly, the encoder proceeds to the swapping step for $L_{\bS}'$ and $R_{\bS}$, where $\operatorname{wt}(L_{\bS}') < n_i'\alpha$ and $\operatorname{wt}(R_{\bS})>n_i'\alpha$, and uses Lemma~\ref{main_lemma} to find the swapping index and the flipping index. 
\vspace{1mm}


After the divide-and-conque procedure, we have an array $B$ of size $(n-c)\times n$ where the weight of each row is at most $f(n)$, and the weight of each column is at most $\floor{mf(n)/n}$. 
    
\subsection{Encoding The Redundancy}
Let ${\rm Re}(n)$ be the sequence obtained by concatenating the binary representations of all the indices $(\Gamma_i, t_i, \tau_i)$ from the divide-and-conquer procedure for a given array of size $m\times n$. For all integers $N$, $1\le N\le \floor{n/2}$, we have: 
\begin{small}
\begin{align*}
&|{\rm Re}(2N)|\le 2 \times |{\rm Re}(N)| + \ceil{\log mN} +2, \text{ and } \\
&|{\rm Re}(2N+1)| \\
&\le |{\rm Re}(N+1)| + |{\rm Re}(N)| + \ceil{\log mN} +\ceil{\log (N+1)}+2. 
\end{align*}
\end{small}
Via mathematical induction, we can show that the size $|{\rm Re}(n)| \le n(6+\log m) < n(6+\log n)$ bits. 

The encoder then encodes the ${\rm Re}(n)$ into an array $C$ of size $c \times n$. 
For simplicity, we assume that $\log n$ is integer, and $\operatorname{Re}(n)=x_1 x_2 \ldots x_{n(\log n+6)}$. We fill the bits $\operatorname{Re}(n)$ into $C$ as follows: in every row (or column) there are at least $(\ceil{n/f(n)}-1)$ 0's between two consecutive bits in $\operatorname{Re}(n)$. Formally, for $1 \leqslant i \leqslant n$, $i=i_0 \times \ceil{n/f(n)}+j$ for some integers $i_0, j$, where $1\le j\le \ceil{n/f(n)}$, we set the $i$th column: 
\begin{align*}
C^i = &(0^{j-1} x_{(\log n+6) (i-1) + 1} 0^{\lceil\frac{n}{f(n)}\rceil-j}) \ldots \\
&\ldots (0^{j-1} x_{(\log n +6) i}0^{\lceil\frac{n}{f(n)}\rceil-j}),
\end{align*}
where $0^m$ denotes $m$ consecutive 0's. Clearly, 
\begin{itemize}
\item {\bf Every row of $C$ is $f$-bounded.} The weight of every $\ceil{n/f(n)}$ consecutive bits is at most one, and hence, the weight of each row of $C$ is at most $n/\ceil{n/f(n)} \le f(n)$.
\item {\bf Every column of $C$ has weight at most $\floor{cf(n)/n}$.} Similarly, since the weight of every $\ceil{n/f(n)}$ consecutive bits is at most one, the weight of every column of size $c$ of $C$ is at most $c/\ceil{n/f(n)} \le cf(n)/n$. Since the weight is an integer number, it is then less than $\floor{cf(n)/n}$. 
\end{itemize}

Finally, the encoder outputs the concatenation of $B$ and $C$, which is an array $D$ of size $n\times n$. Clearly, every row of $D$ is $f$-bounded, while the weight of every column of $D$ is at most 
\begin{align*}
\floor{mf(n)/n} + \floor{cf(n)/n} &\le \floor{mf(n)/n+cf(n)/n} \\
&=\floor{f(n)} \le f(n).
\end{align*}

We illustrate the idea of the divide-and-conquer procedure in Figure~\ref{example}. 
We now discuss the efficiency of our encoder. Let ${\bf c}_f$ denote the channel capacity of the 1D constraint, i.e. given a function $f$, the weight of every codeword of length $n$ is at most $f(n)$. Suppose that ${\bf r}_{\phi}, {\bf Re}_\phi$ are the rate and the redundancy of the capacity-approaching 1D encoding map $\phi$, respectively. Here, ${\bf r}_{\phi}=1-{\bf Re}_\phi/n$ and $\lim_{n\to \infty} {\bf r}_{\phi}={\bf c}_f$. 

\begin{theorem}\label{main-theorem}
Our encoder provides capacity-approaching 2D codes for all $f=\omega(\log n)$. 
In addition, the channel capacity of the 2D constraint is equal to the channel capacity of the 1D constraint. In other words, we have ${\bf cap}_f \equiv {\bf c}_f$. 
\end{theorem}

\begin{proof}
It is easy to verify that the total redundancy for encoding an $n \times n$ array is $(n-c) {\bf Re}_{\phi}+ c n $ (bits), where $c=\ceil{n/f(n)} \lceil\log n + 6\rceil$. For simplicity, we assume that $c=n(\log n+6)/f(n)$ is an integer. The rate of our encoder, denoted by ${\bf r}_{{\enc}_f}$, is then 
\begin{align*}
{\bf r}_{{\enc}_f} &= \frac{n^2-(n-c) {\bf Re}_{\phi}-cn}{n^2} \\
&=  \frac{n^2- n {\bf Re}_{\phi} + c{\bf Re}_\phi-cn}{n^2} \\
&= \Big( 1-\frac{{\bf Re}_{\phi}}{n} \Big) + \frac{{\bf Re}_\phi(\log n+6)}{nf(n)}-\frac{\log n+6}{f(n)} \\
&= {\bf r}_{\phi}+ \frac{{\bf Re}_\phi(\log n+6)}{nf(n)}-\frac{\log n+6}{f(n)}. 
\end{align*}
For all $f=\omega(\log n)$, we have $\lim_{n\to\infty} \frac{\log n+6}{f(n)} =0$. Therefore, $\lim_{n\to\infty} {\bf r}_{{\enc}_f} = \lim_{n\to \infty} {\bf r}_{\phi}={\bf c}_f$. We then conclude that our encoder provides capacity-approaching 2D codes for all $f=\omega(\log n)$, and ${\bf cap}_f \equiv {\bf c}_f$. 
\end{proof}

  \begin{figure}[t!]
    \centering
    \includegraphics[width=0.37\textwidth]{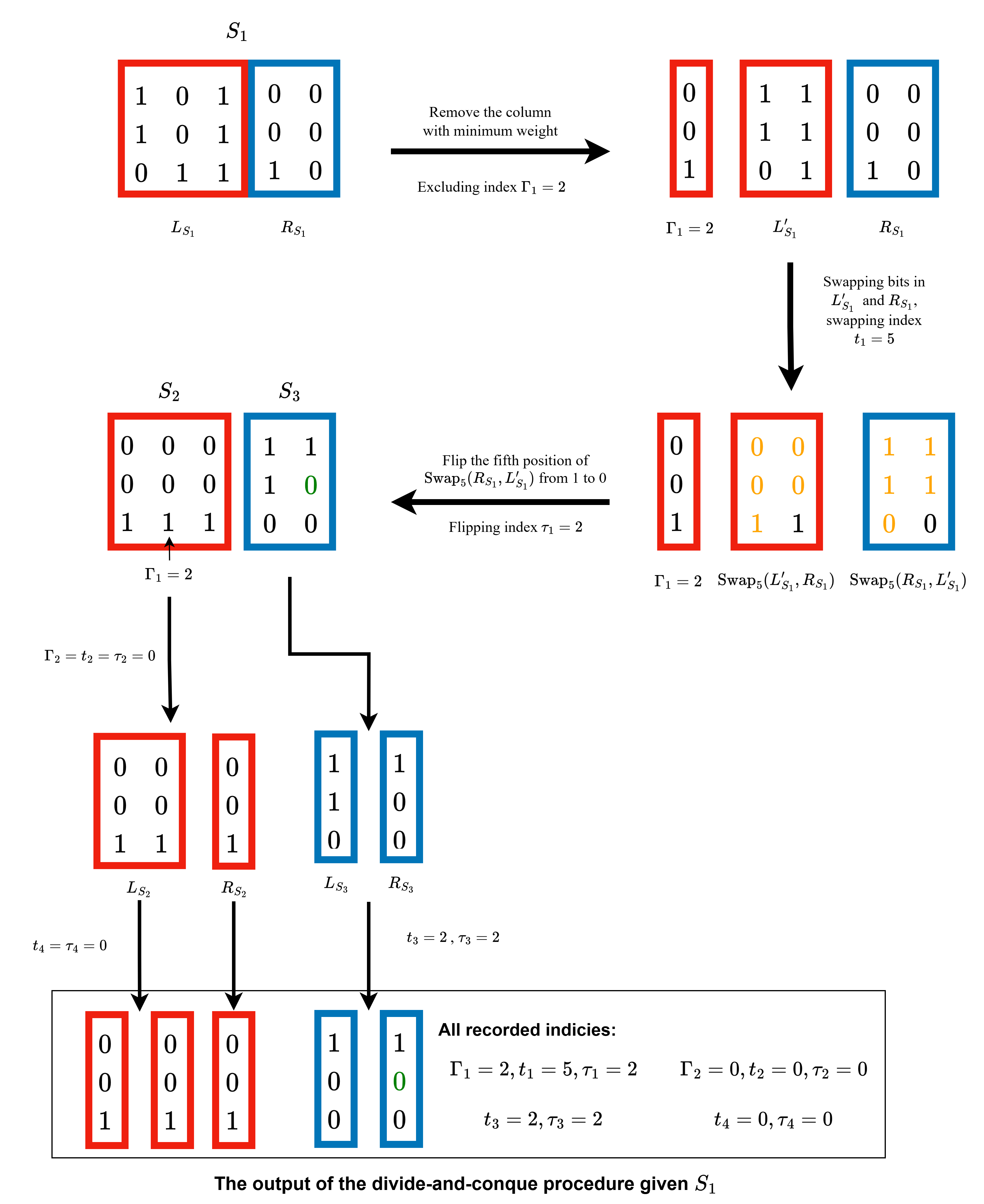}
   \caption{Given an array $\bS_1$ consisting of $5$ columns, and $\alpha = 1.5$. It is uniquely encodable and decodable given all indices $(\Gamma_i, t_i, \tau_i)$.}\label{example}
    \end{figure} 

\begin{remark}
It is important to highlight that our approach does not require the constraint $cf(n)/n \in \mathbb{N}$ (as required in \cite{nguyen2023two}). Therefore, it simplifies the parameter selection process while preserving both theoretical efficiency and practical flexibility.
\end{remark}

\section{Conclusion}
We have presented a universal framework to design 2D codes that guarantee the weight of every row and every column of length $n$ to be at most $f(n)$ for arbitrary function $f(n)$, through efficient encoding and decoding algorithms. 
We have shown that if a design of capacity-approaching 1D codes exists, then our encoder provides capacity-approaching 2D codes for all $f=\omega(\log n)$.




\end{document}